\documentclass[conference,letterpaper]{IEEEtran}

\IEEEoverridecommandlockouts

\usepackage[utf8]{inputenc} 
\usepackage[T1]{fontenc}
\usepackage{url}
\usepackage{ifthen}
\usepackage{cite}
\usepackage[cmex10]{amsmath} 

\usepackage{amssymb}
\usepackage{xcolor}
\usepackage{graphicx}
\newtheorem{proposition}{Proposition}
\usepackage{dirtytalk}
\usepackage{comment}
\usepackage{orcidlink}
\usepackage{algorithm}
\usepackage{algpseudocode}
\usepackage{tikz}
\usetikzlibrary{trees,positioning}
\usepackage{booktabs}
\usepackage[normalem]{ulem}

\algrenewcommand\algorithmicrequire{\textbf{Input:}}
\algrenewcommand\algorithmicensure{\textbf{Output:}}

\newcommand{\cS}{\mathcal{S}}

\newcommand{\bfp}{{\boldsymbol p}}

\renewcommand{\le}{\leqslant}
\renewcommand{\leq}{\leqslant}
\renewcommand{\ge}{\geqslant}

\newcommand{\Cref}[1]{Co\-ro\-lla\-ry\,\ref{#1}}



\outer\def\proclaim #1. #2\par{\medbreak
 \noindent{\bf#1.\enspace}{\sl#2\par}%
 \ifdim\lastskip<\medskipamount \removelastskip\penalty55\medskip\fi}

\newtheorem{theorem}{Theorem}
\newtheorem{definition}{Definition}

\newtheorem{lemma}{Lemma}

\begin{document}
\title{An Additive Approximation Scheme for Generating Dyadic Codings for the
Outputs of an LLM\vspace{-0.8ex}}

\author{%
  \IEEEauthorblockN{Daniella Bar-Lev\IEEEauthorrefmark{1},
                    Farzad Farnoud\IEEEauthorrefmark{2},
                    and Ryan Gabrys\IEEEauthorrefmark{3}\IEEEauthorrefmark{4}}
\thanks{\IEEEauthorrefmark{1} Department of Mathematics, University of Zurich, Switzerland;
    \IEEEauthorrefmark{2}Department of Electrical and Computer Engineering and the Department of Computer Science UVA, USA;
  \IEEEauthorrefmark{3}Calit2 at the UCSD, USA; \IEEEauthorrefmark{4}Naval Information Warfare Center Pacific, USA.}
  \thanks{Emails: daniella.bar-lev@math.uzh.ch, farzad@virginia.edu, rgabrys@ucsd.edu, ryan.gabrys.civ@us.navy.mil}
  \thanks{
  The work of D. Bar-Lev was supported by Schmidt Sciences and by the Swiss National Science Foundation under grant number 212865. 
  }
  \vspace{-8ex}}

\maketitle

\begin{abstract}
We study the problem of approximating a discrete probability distribution, such as the next-token distribution of a large language model, by a dyadic distribution induced by a binary tree under encoding rate constraints. The objective is to partition the support of the distribution and assign dyadic probabilities to minimize total variation distance while achieving a prescribed rate. We formulate this task as a tree-based partitioning problem and develop a polynomial-time additive approximation scheme for the rate-constrained setting in the constant-rate regime. Our results provide provable guarantees for near-optimal dyadic approximations and, as an application, yield a principled framework for LLM-based steganography, where the rate maps to bits of hidden information embedded per token and the total variation bound controls statistical detectability.
\end{abstract}

\section{Introduction}
Large language models (LLMs) generate text by sampling, at each step, from a categorical distribution over a large token vocabulary. Many downstream applications, such as LLM steganography, constrained generation, and structured discretization, need to replace this categorical distribution by a tractable approximation that interfaces cleanly with binary coding. Dyadic distributions, those whose probabilities are negative powers of two, fit this need. They correspond exactly to the leaf masses of full binary trees, where the root-to-leaf path of each leaf is itself a bitstring. For steganography, this gives a direct one-to-one correspondence between leaves and hidden messages. Specifically, sampling a token from the dyadic distribution is  equivalent to transmitting the bitstring labeling its leaf, with the leaf mass $2^{-d}$ determining how often that $d$-bit message is sent. This makes prefix-free coding, sampling, and arithmetic-coding manipulations straightforward.

This paper studies the resulting approximation problem. Given a categorical distribution $\bfp$ over $n$ tokens, we look for a full binary tree together with a surjective assignment of tokens to its leaves, so that the induced dyadic distribution is close to $\bfp$ in some statistical sense and the resulting encoding has rate at least $R$. We measure closeness by total variation distance because, in the LLM steganography setting that motivates this paper, it directly bounds an adversary's advantage in detecting the hidden channel via binary hypothesis testing between the cover and stego distributions. This is in contrast to much of the existing LLM steganography literature, which optimizes Kullback--Leibler divergence or cross-entropy surrogates~\cite{Dai2019,Ziegler2019,Huang2024ODStega}.

We call this the \emph{Tree-Partitioning Problem} (TPP). The number of leaves $L$ is induced by the tree rather than fixed in advance, and the rate constraint couples the tree shape to the partition through Kraft's inequality~\cite{CoverThomas2006}. Even very restricted instances of TPP --- e.g., depth-$1$ trees with $L=2$ --- are NP-hard via reduction from the partition problem, so we cannot hope for exact polynomial-time solutions. We instead build a polynomial-time additive approximation scheme, drawing on the classical polynomial-time approximation scheme (PTAS) toolkit for structured knapsack problems~\cite{CapraraMSSP2000,MartelloToth1990,Kellerer2004}.

We work throughout in the \emph{constant-rate regime}, where the target rate $R$ is fixed independently of $n$. This regime fits the LLM steganography setting, where the encoder needs the channel to carry at least some minimum number of bits per token but is otherwise free to optimize statistical fidelity. By contrast, regimes in which $R$ or $L$ scales with $n$ degenerate either into trivial near-symbol-wise constructions or into substantially more complex enumeration problems. We comment briefly on these in the discussion.

We make three contributions. First, we introduce the TTP as a unified formulation of dyadic approximation under joint divergence and rate constraints, making explicit the roles of tree shape, leaf assignment, and Kraft's inequality. Second, for the rate-constrained TPP (Problem~1), we give a polynomial-time additive approximation scheme with running time $O(n)\cdot \exp(O(\log(1/\varepsilon)/\varepsilon))$
and total-variation guarantee $\mathrm{OPT}+12\varepsilon$ in the constant-rate regime. Third, we position variable-length dyadic encodings and total variation distance as a natural but underexplored foundation for LLM steganography.

To design this approximation scheme, our algorithm builds on a standard PTAS toolkit for knapsack-type partitioning and load-balancing problems~\cite{Sahni1975,IbarraKim1975,Lawler1979,Kellerer2004}. Specifically, we synthesize several techniques from this literature:
mall-/large-item classification, aggregation of negligible items,  depth truncation, structural monotonicity of optimal solutions, rounding onto a finite lattice, and dynamic programming over compressed load vectors.
For a fixed height vector ${\bf H}$, the assignment subproblem has the flavor of target-load
partitioning. More specifically, the token probabilities play the role of item sizes, leaves play the
role of bins, and the dyadic masses $2^{-h_j}$ are the target loads. This
connects the fixed-tree subproblem to multiway number partitioning,
scheduling/load-balancing, and bin-assignment objectives based on absolute
deviation from prescribed targets
\cite{Korf2009,Alon1998,Glover1989}. The full TPP, however, is not a direct
instance of these models. Its target loads are not fixed in advance and instead they are
induced by a binary tree that must be optimized jointly with the partition.
Consequently, the height vector must simultaneously determine dyadic masses,
satisfy Kraft's equality, and meet the rate constraint
$\sum_j 2^{-h_j}h_j\ge R$. Thus, while the fixed-tree assignment problem is
knapsack-like, the main technical work is to adapt these ideas to a coupled
partition-and-prefix-coding problem, a structure absent from standard
multiple-knapsack, multiple-subset-sum, and load-balancing formulations
\cite{CapraraMSSP2000,ChekuriKhanna2006}.

Section~\ref{sec:problem} defines the TPP, introduces two constrained variants, and states the NP-hardness result that motivates approximation. Section~\ref{sec:setup} specializes to the constant-rate regime by grouping items by mass and imposing three instance-level conditions. Section~\ref{sec:alg} contains the technical core. Here, we apply four reductions (truncation, grouping, dynamic programming, and feasibility repair) that together prove the main theorem. Section~\ref{sec:discussion} discusses the relevance of the assumptions for LLM-scale supports, the dual formulation, and other scaling regimes.

\section{Definitions and Problem Statement}\label{sec:problem}

Consider a categorical distribution $\bfp=\left( p_1, p_2, \ldots, p_n \right)$ defined over a support of size $n$, representing the output token distribution of an LLM. Without loss of generality, we assume the probabilities are sorted in non-increasing order: $p_1\ge p_2\ge \ldots \ge p_n$.

Our objective is to choose a full binary tree together with a partition of the support onto its leaves. For any candidate full binary tree, let $L$ denote its number of leaves. Thus $L$ is not fixed in advance; rather, it is induced by the tree under consideration. We assume throughout that every leaf is used, i.e., empty leaves are not permitted.

\begin{definition}[Partition Function]
Fix a candidate full binary tree with $L$ leaves. Given integers $n\ge L$, let 
$
\pi:[n]\twoheadrightarrow[L]
$
be a surjective partition function, so that every leaf is used. For each $j\in [L]$ we define\vspace{-0.75ex}
\[
\cS_j= \left\{ i\in [n]: \pi(i) = j \right\}\vspace{-0.75ex}
\]
to be the set of token indices assigned to leaf $j$. Note that 
$\cS_j\ne\varnothing$ for all $j$. The partition probability is denoted\vspace{-0.75ex}
\[
\Pr(\cS_j) = \sum_{i\in \cS_j} p_i.\vspace{-0.75ex}
\]
We represent the partition through the ordered sequence\vspace{-0.75ex}
\[
{\bf S} = \left(\cS_1, \cS_2, \ldots, \cS_L\right).\vspace{-0.75ex}
\]
\end{definition}

\begin{definition}[Height Function]
Fix a candidate full binary tree with $L$ leaves. Let $h:[L]\to[H]$ be the height function that returns the depth of leaf $j$
in the encoding tree, where $H$ denotes the maximum depth of the tree. We represent this function through the ordered vector\vspace{-0.75ex}
\[
{\bf H}= (h_1,h_2,\ldots, h_L),\vspace{-0.75ex}
\]
where $h_j=h(j)$ denotes the height of leaf $j$.
\end{definition}
\begin{definition}[Information Rate] For a given height function $h$, we define the information rate of the encoding tree distribution as\vspace{-0.75ex}
\[
\mathfrak{R}({\bf S},{\bf H}) = \sum_{j=1}^L 2^{-h_j}\cdot h_j.\vspace{-0.75ex}
\]
\end{definition}
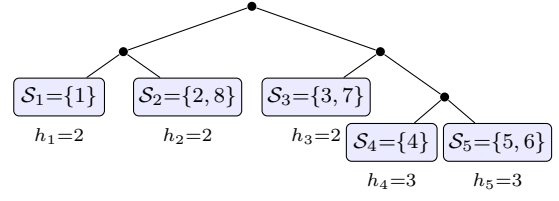
\begin{definition}[Divergence Measure] For a given partition function $\pi$ and height function $h$, we define the total variation distance\footnote{The variation distance is typically defined as $\frac{1}{2} \sum_{j=1}^L \left|2^{-h_j} - \Pr(\cS_j)\right|$, but we chose to omit the $\frac{1}{2}$ coefficient for brevity since the  minimization of both formulations is the same.} between the partition distribution and the encoding tree distribution as\vspace{-0.75ex}
\[
\mathfrak{D}({\bf S},{\bf H}) = \sum_{j=1}^L \left|2^{-h_j} - \Pr(\cS_j)\right|.\vspace{-0.75ex}
\]
\end{definition}

Figure~\ref{fig:example} shows a small running example with $n=8$ 
and $L=5$ 
that illustrates the partition function $\pi$, the partition ${\bf S}$,
the height vector $\bf H$, the divergence $\mathfrak{D}$, and the rate $\mathfrak{R}$. 

\begin{figure}[!t]
\centering
\begin{tikzpicture}[
  every node/.style={font=\footnotesize},
  leaf/.style={draw, rounded corners=2pt, fill=blue!8,
               minimum width=12mm, minimum height=5mm, inner sep=2pt},
  inn/.style={circle, fill=black, inner sep=1.2pt},
  level distance=6mm,
  level 1/.style={sibling distance=34mm},
  level 2/.style={sibling distance=17mm},
  level 3/.style={sibling distance=14mm},
  edge from parent/.style={draw, thin}
]
\node[inn] {}
  child {node[inn] {}
    child {node[leaf] (l1) {$\cS_1{=}\{1\}$}}
    child {node[leaf] (l2) {$\cS_2{=}\{2,8\}$}}
  }
  child {node[inn] {}
    child {node[leaf] (l3) {$\cS_3{=}\{3,7\}$}}
    child {node[inn] {}
      child {node[leaf] (l4) {$\cS_4{=}\{4\}$}}
      child {node[leaf] (l5) {$\cS_5{=}\{5,6\}$}}
    }
  };
\node[below=0.3mm of l1, font=\scriptsize] {$h_1{=}2$};
\node[below=0.3mm of l2, font=\scriptsize] {$h_2{=}2$};
\node[below=0.3mm of l3, font=\scriptsize] {$h_3{=}2$};
\node[below=0.3mm of l4, font=\scriptsize] {$h_4{=}3$};
\node[below=0.3mm of l5, font=\scriptsize] {$h_5{=}3$};
\end{tikzpicture}

\vspace{2pt}

{\footnotesize
\renewcommand{\arraystretch}{1.0}
\begin{tabular}{@{}c@{\;\;}c@{\;\;}c@{\;\;}c@{\;\;}c@{}}
\toprule
$j$ & $\cS_j$ & $\Pr(\cS_j)$ & $2^{-h_j}$ & $|2^{-h_j}-\Pr(\cS_j)|$ \\
\midrule
1 & $\{1\}$    & $0.30$ & $0.250$ & $0.050$ \\
2 & $\{2,8\}$  & $0.22$ & $0.250$ & $0.030$ \\
3 & $\{3,7\}$  & $0.20$ & $0.250$ & $0.050$ \\
4 & $\{4\}$    & $0.12$ & $0.125$ & $0.005$ \\
5 & $\{5,6\}$  & $0.16$ & $0.125$ & $0.035$ \\
\bottomrule
\end{tabular}}

\caption{Running example with $n=8$ tokens, $L=5$ leaves, and probabilities $\bfp=(0.30,\,0.20,\,0.15,\,0.12,\,0.10,\,0.06,\,0.05,\,0.02)$. The partition function $\pi:[8]\twoheadrightarrow[5]$ assigns $1{\mapsto}1$, $\{2,8\}{\mapsto}2$, $\{3,7\}{\mapsto}3$, $4{\mapsto}4$, and $\{5,6\}{\mapsto}5$, inducing the sets $\cS_j$ shown at the leaves. The height vector ${\bf H}=(2,2,2,3,3)$ satisfies Kraft's equality $\sum_j 2^{-h_j}=1$. Summing the last column of the table gives $\mathfrak{D}({\bf S},{\bf H})=0.17$, and the rate is $\mathfrak{R}({\bf S},{\bf H})=\sum_j 2^{-h_j}h_j=3\cdot \tfrac14\cdot 2+2\cdot\tfrac18\cdot 3=2.25$.\vspace{-3.75ex}}
\label{fig:example}
\end{figure}

\subsection{The Tree-Partitioning Problem (TPP)}
In its basic form, the TPP seeks a partition function $\pi$ and height function~$h$ that seeks to both minimize the divergence $\mathfrak{D}({\bf S}, {\bf H})$ and maximize the information rate $\mathfrak{R}({\bf S}, {\bf H})$ subject to the constraint that the heights correspond to a valid full binary tree so that\vspace{-0.75ex}
\begin{align}
\label{eq:tree}
\sum_{j=1}^L 2^{-h_j}=1.\vspace{-0.75ex}
\end{align}
Unless stated otherwise, we assume all height functions studied in this work correspond to a valid full binary tree and thus satisfies (\ref{eq:tree}) according to the Kraft inequality \cite{CoverThomas2006}.

We now introduce two constrained variants of the
TPP, which incorporate rate constraints and rate optimization.

\textit{\textbf{Problem 1: Rate-Constrained
TPP.}}
Given $\bfp$ and a fixed encoding rate $R>0$ find $({\bf S}, {\bf H})$ that minimizes
$
\mathfrak{D}({\bf S},{\bf H})$
subject to\vspace{-0.75ex}
\begin{align}
\mathfrak{R}({\bf S},{\bf H}) \ge R. \label{eq:rate1}\vspace{-0.75ex}
\end{align}

\textit{\textbf{Problem 2: Divergence-Constrained
TPP.}} Given $\bfp$ and a fixed tolerance $\Delta>0$ find  $({\bf S}, {\bf H})$ that maximizes the information rate $\mathfrak{R}({\bf S},{\bf H})$, subject to\vspace{-0.75ex}
\begin{align}
\mathfrak{D}({\bf S}, {\bf H}) \le \Delta. \label{eq:KLF}\vspace{-0.75ex}
\end{align}

The body of this paper focuses on Problem~1 in the constant-rate regime, which is the relevant regime for steganography. One typically imposes a minimum rate and optimizes detectability, measured here through total variation distance~\cite{Cachin1998,Fridrich2009,Bloch2016}. Problem~2 and regimes where $R$ or $L$ scales with $n$ are discussed in Section~\ref{sec:discussion}.

As even very restricted instances of the TPP are computationally hard, we cannot hope for exact polynomial-time solutions. We therefore turn to approximation.

\begin{lemma}\label{lem:np-hard}
The TPP is NP-hard, even when the tree depth is fixed to $1$ (corresponding to $L=2$ partitions).
\end{lemma}
\begin{IEEEproof}
A direct reduction from the partition problem~\cite{GareyJohnson1979} to TPP with $L=2$ and depth $1$.
\end{IEEEproof}

\section{Constant-Rate Regime: Setup}\label{sec:setup}

We work in the constant-rate regime, where $R$ is fixed independently of $n$ and $\varepsilon\to 0$ controls the additive approximation error. This section sets up the instance conditions used by the approximation algorithm. We first classify items as small or large according to their probability mass, which ensures that only constantly many items can be large. We then state three numerical assumptions on the support size and on $\varepsilon$ for the seeding and repair steps in Section~\ref{sec:alg}.

\subsection{Small and Large Item Classification}
We split the support by probability mass at the threshold $\theta\triangleq\varepsilon\cdot 2^{-R/\varepsilon}$. The {\bfseries\itshape small item set} is\vspace{-0.75ex}
\begin{align}\label{eq: small}
\cS_{\text{small}} \triangleq \{i\in [n]: p_i\le \theta\},
\vspace{-0.75ex}\end{align}
and the {\bfseries\itshape large item set} is $\cS_{\text{large}} \triangleq [n] \setminus \cS_{\text{small}}$. Each large item has mass exceeding the constant $\theta$, so the total mass constraint $\sum_i p_i = 1$ forces $|\cS_{\text{large}}|\le 1/\theta = O(1)$ for fixed $R, \varepsilon$, and consequently $|\cS_{\text{small}}|\ge n - O(1)$.

\subsection{Three Assumptions on the Instance}\label{sec:assumptions}

The algorithm of Section~\ref{sec:alg} operates under the following three numerical assumptions on the instance $(\bfp,n,R,\varepsilon)$. Their role is to ensure that the support is large enough to support a depth-$R$ encoding tree and contains enough small items for the technical operations performed in Section~\ref{sec:alg}.
\medskip

\noindent\textbf{Assumption 1.} $n\ge 2^{R}$. The minimum vocabulary needed to support a tree with rate at least $R$ via leaves at depth $\ge R$.

\noindent\textbf{Assumption 2.} $|\cS_{\text{small}}| \ge 2^{R/\varepsilon} + 1/\varepsilon$. Minimal 
number of small items, 
required by Lemmas~\ref{lem:repair} and~\ref{lem:seed} below.

\noindent\textbf{Assumption 3.} $\varepsilon$ is small enough that $2^{R/\varepsilon}\ge 1/\varepsilon$. Automatic for all sufficiently small $\varepsilon$ when $R$ is constant.

\medskip

\noindent We interpret each candidate solution as a dyadic tree whose leaves correspond to the partition subsets $\{\cS_j\}_{j=1}^L$.

\section{Algorithm and Analysis}\label{sec:alg}
\subsection{Roadmap}\label{sec:roadmap}

The hardness of TPP (Lemma~\ref{lem:np-hard}) rules out exact polynomial-time solutions and motivates our additive approximation.
The search space --- arbitrary full binary trees together with arbitrary partitions of $[n]$ onto their leaves --- is far too large to enumerate. Two of the four reductions shrink this search space \emph{combinatorially}, and one shrinks it \emph{numerically}. The fourth restores any feasibility lost in the process.

\medskip

\noindent\textbf{(i) Truncation.} For an additive cost of zero in divergence, it suffices to consider trees of depth at most $d\triangleq\log_2(1/\varepsilon)$. The set of such tree shapes is finite, with size $\exp(O(1/\varepsilon))$. This reduction, made precise in Lemma~\ref{lem:truncation}, replaces the search over heights by enumeration over a constant number of bounded-depth tree shapes.

\noindent\textbf{(ii) Grouping.} Tokens with mass below $\varepsilon^2$ can be aggregated into ``blocks'' of mass between $\varepsilon^2$ and $2\varepsilon^2$ at a divergence cost of $4\varepsilon$. After grouping, the number of objects to assign (the original heavy items together with the blocks) is at most $1/\varepsilon^2+1$, a constant. This step (Lemma~\ref{lem:blocking}) reduces the per-tree assignment problem to constant size.

\noindent\textbf{(iii) Atomic assignment via dynamic programming.} For each fixed bounded-depth tree, the resulting bounded-size assignment problem is solved by a dynamic program over a discretized leaf-mass lattice, returning an assignment within $2\varepsilon$ of the divergence-optimal blocked assignment in time $O(n)\cdot\exp(O(\log(1/\varepsilon)/\varepsilon))$ (Proposition~\ref{prop:large-item-dp}).

\noindent\textbf{(iv) Feasibility repair.} The depth-truncated trees produced in step (i) need not satisfy $\mathfrak{R}\ge R$. Whenever a bounded-depth candidate has a leaf at depth exactly $d$, that leaf can be replaced by a complete subtree of $T_L=2^{R/\varepsilon}$ uniform sub-leaves seeded with reserved small items. This restores rate at additive divergence cost $4\varepsilon$ (Lemma~\ref{lem:repair}). 

\medskip

Combining 
(i)-(iv) along with a leaf-seeding step that potentially includes a divergence cost of $2 \epsilon$, we obtain Theorem~\ref{thm:main}: an algorithm with running time linear in $n$ for fixed $\varepsilon$ that returns a feasible solution within $\mathrm{OPT}+12\varepsilon$ in total variation distance.
Throughout, it will be convenient to denote by\vspace{-0.75ex}
\[
\mathrm{OPT}({\bf H})
\triangleq
\min_{\bf S}\,\mathfrak D({\bf S},{\bf H})\vspace{-0.75ex}
\]
the divergence-optimal partition cost for a height vector ${\bf H}$, where the minimum is over surjective partitions ${\bf S}$ for ${\bf H}$.

\subsection{Truncation Reduction}\label{sec:truncation}

The first reduction bounds the depth of the trees we have to enumerate. Any feasible solution can be replaced, with no increase in divergence, by one of depth at most $d = \log_2(1/\varepsilon)$. The set of full binary tree shapes of bounded depth is finite (with size $\exp(O(1/\varepsilon))$), so this turns the outer search into a finite enumeration at no cost in approximation quality. Truncation can, however, lower the rate, since shallower leaves contribute less to $\mathfrak{R}$. Step (iv) (Section~\ref{sec:repair}) recovers the lost rate by re-expanding a depth-$d$ leaf into a uniform subtree, which is why the lemma additionally guarantees the existence of such a leaf.

\begin{lemma}[Truncation reduction]\label{lem:truncation}
Let $({\bf S}^*,{\bf H}^*)$ be any feasible solution to Problem~1, and set $d=\log_2(1/\varepsilon)$. There exists a height vector $\overline{\bf H}$ of maximum depth at most $d$ such that\vspace{-0.75ex}
\[
\mathrm{OPT}(\overline{\bf H})
\le
\mathfrak{D}({\bf S}^*,{\bf H}^*).\vspace{-0.75ex}
\]
Moreover, if $\max_j h^*_j>d$ then $\overline{\bf H}$ has at least one leaf at depth exactly $d$.
\end{lemma}
\begin{IEEEproof}
If $\max_j h^*_j\le d$, take $\overline{\bf H}={\bf H}^*$ and the claim is immediate.

Otherwise, construct $\overline{\bf H}$ from ${\bf H}^*$ by truncating at depth $d$: every internal node of ${\bf H}^*$ at depth $d$ becomes a leaf of $\overline{\bf H}$. Define $\overline{\bf S}$ from ${\bf S}^*$ by mechanically merging all partition sets that lie below depth $d$ into the corresponding new leaf. Each individual merge of two siblings at depth $h\ge d+1$ replaces the two contributions $|2^{-h}-p|+|2^{-h}-p'|$ in $\mathfrak{D}({\bf S}^*,{\bf H}^*)$ by the single contribution $|2^{-(h-1)}-(p+p')|$, which by the triangle inequality is no larger. Iterating these merges,
\[
\mathfrak D(\overline{\bf S},\overline{\bf H})\le \mathfrak D({\bf S}^*,{\bf H}^*),
\]
and hence
\[
\mathrm{OPT}(\overline{\bf H})\le \mathfrak D(\overline{\bf S},\overline{\bf H})\le \mathfrak D({\bf S}^*,{\bf H}^*).
\]
Since $\max_j h^*_j>d$, the truncation creates at least one leaf at depth exactly $d$.
\end{IEEEproof}

\subsection{Grouping Tiny Items}\label{sec:blocking}

After truncation, every candidate height vector has at most $1/\varepsilon$ leaves. The remaining task, divergence-optimally assigning $n$ items to those leaves, still depends on $n$, which for an LLM vocabulary is in the tens of thousands. The next reduction shrinks $n$ to a constant: items of mass below $\varepsilon^2$ are aggregated into ``blocks'' of mass between $\varepsilon^2$ and $2\varepsilon^2$, at a divergence cost of $4\varepsilon$. After blocking, the number of \emph{atomic units} --- original heavy items together with blocks --- is at most $1/\varepsilon^2+1$, independent of $n$. Section~\ref{sec:large-item-assignment} solves the resulting bounded-size assignment problem by dynamic programming.

\begin{lemma}[Blocking tiny items]\label{lem:blocking}
Fix a candidate height vector $\overline{\bf H}$ of maximum depth
$d=\log_2(1/\varepsilon)$, and let $L$ denote its number of leaves. Let\vspace{-0.75ex}
\[
\mathcal T \triangleq \{i\in[n]: p_i<\varepsilon^2\}.\vspace{-0.75ex}
\]
Form blocks $\mathcal B_1,\ldots,\mathcal B_M$ by greedily grouping the
items of $\mathcal T$ so that each completed block has total probability
in $[\varepsilon^2,2\varepsilon^2)$, with at most one final residual
block $\mathcal U$ satisfying
$
\Pr(\mathcal U)<\varepsilon^2.
$
We call the original heavy items (those $i\in[n]$ with $p_i\ge\varepsilon^2$) together with the blocks $\mathcal B_1,\ldots,\mathcal B_M,\mathcal U$ the \emph{atomic units} of the blocked instance.

Then for every partition ${\bf S}$ for $\overline{\bf H}$, there
exists another partition $\widehat{\bf S}$ for the same height
vector in which all elements of each block $\mathcal B_i$ and of $\mathcal U$ are mapped to the same leaf, such that
\[
\left|
\mathfrak D(\widehat{\bf S},\overline{\bf H})
-
\mathfrak D({\bf S},\overline{\bf H})
\right|
\le 4\varepsilon .
\]
\end{lemma}

\begin{IEEEproof}
Consider a partition ${\bf S}$ for $\overline{\bf H}$. For each leaf $j\in[L]$, let
\[
x_j \triangleq \Pr(\mathcal T\cap \mathcal S_j)
\]
 denote the total mass of the tiny items assigned to leaf $j$ in the
 partition ${\bf S}$. We now replace this tiny-item assignment with an
 assignment of whole blocks.

 Process the leaves in order $1,2,\ldots,L-1$. For leaf $j$, assign whole
 completed blocks greedily until the accumulated block mass first reaches
 or exceeds $x_j$, or until no unassigned completed blocks remain. Let
 $\widehat{x}_j$ denote the resulting total block mass assigned to leaf
 $j$. Assign all remaining completed blocks together with the residual
 block $\mathcal U$ to the final leaf $L$, and let $\widehat{x}_L$ denote
 the total mass thereby assigned. Set
 $\Delta_j\triangleq \widehat{x}_j-x_j$.

 Every tiny item is reassigned somewhere, so the total tiny-item mass is
 preserved:
 \[
 \sum_{j=1}^L \widehat{x}_j = \Pr(\mathcal T) = \sum_{j=1}^L x_j,
 \qquad\text{i.e.,}\qquad
 \sum_{j=1}^L \Delta_j = 0.
 \]
 We bound $\sum_j|\Delta_j|$ via this conservation. At each leaf $j<L$
 where the greedy met the target $x_j$, the overshoot is bounded by the
 mass of the last block assigned, so $\Delta_j\in[0,2\varepsilon^2)$; at
 every other leaf $j<L$ the greedy ran out of completed blocks, giving
 $\Delta_j\le 0$. At leaf $L$ there are two cases. If greedy did not run out at any $j<L$,
 then every $\Delta_{j<L}\ge 0$ and conservation
 $\sum_{j=1}^L\Delta_j=0$ forces
 $\Delta_L=-\sum_{j<L}\Delta_j\le 0$, so $\Delta_L$ contributes nothing to
 the positive part. If greedy did run out at some $j_0<L$, then no
 completed blocks remain by leaf $L$, so
 $\widehat{x}_L=\Pr(\mathcal U)<\varepsilon^2$ and hence
 $\Delta_L<\varepsilon^2$ in this case. Summing the positive part,
 \[
 \sum_{j=1}^L \max(\Delta_j,0) \;<\; (L-1)\cdot 2\varepsilon^2+\varepsilon^2 \;<\; 2L\varepsilon^2.
 \]
 Because $\sum_j\Delta_j=0$, the same bound applies to the negative part,
 so
 \[
 \sum_{j=1}^L |x_j-\widehat{x}_j| \;=\; \sum_{j=1}^L|\Delta_j| \;<\; 4L\varepsilon^2 \;\le\; 4\varepsilon,
 \]
 where the last inequality uses $L\le 2^d=1/\varepsilon$.

 Now let $y_j$ denote the total mass of all items outside $\mathcal T$
 assigned to leaf $j$. Then the original and modified leaf masses are
 $y_j+x_j$ and $y_j+\widehat{x}_j$, respectively. Using the inequality
 $\bigl||a-c|-|b-c|\bigr|\le |a-b|$, we obtain
 \begin{align*}
 \left|
 \mathfrak D(\widehat{\bf S},\overline{\bf H})
 -
 \mathfrak D({\bf S},\overline{\bf H})
 \right|
 &\le
 \sum_{j=1}^L
 \left|
 (y_j+\widehat{x}_j)-(y_j+x_j)
 \right|\\
 &=\sum_{j=1}^L |x_j-\widehat{x}_j|\\
 &\le 4\varepsilon.
 \end{align*}
\end{IEEEproof}
\medskip

After applying Lemma~\ref{lem:blocking}, the algorithm works with a collection of \emph{atomic units} consisting of (i) each original item of mass $\ge \varepsilon^2$ and (ii) each block produced from items of mass below $\varepsilon^2$. Every atomic unit except possibly the final residual block has mass at least $\varepsilon^2$, so their total count, denoted by $K$, is bounded by $\varepsilon^{-2}+1$,
a constant for fixed $\varepsilon$.

\subsection{Atomic Assignment via Dynamic Programming}\label{sec:large-item-assignment}

After truncation and blocking, what is left for each candidate height vector $\overline{\mathbf H}$ is a finite combinatorial problem. Assign $K\le 1/\varepsilon^2+1$ atomic units to $L\le 1/\varepsilon$ leaves so as to minimize divergence. Even in this finite form, brute force is exponential in $1/\varepsilon$ 
as there are $L^K \le (1/\varepsilon)^{1/\varepsilon^2}$ possible assignments.

Divergence depends on the assignment only through the resulting leaf-mass vector, not on which atomic unit lands where. Discretizing each atomic unit's mass to the nearest multiple of $\delta\triangleq\varepsilon^3/2$ collapses the state space to a polynomial-size lattice of leaf-mass vectors, on which the assignment problem reduces to a standard knapsack-style dynamic program.

\paragraph{Discretization}
Let the atomic units be indexed by $u\in[K]$, where $K\le \varepsilon^{-2}+1$
and let $q_u$ denote the probability mass of atomic unit $u$. Set\vspace{-0.75ex}
\[
\delta \triangleq \frac{\varepsilon^3}{2},
\quad
w_u \triangleq \left\lfloor \frac{q_u}{\delta}\right\rfloor
\in \mathbb Z_{\ge 0},
\quad
\widehat q_u \triangleq \delta\, w_u, \quad W \triangleq \sum_{u=1}^K w_u.\vspace{-0.75ex}\]
Since the atomic units partition $[n]$, $\sum_{u=1}^K q_u =1$, so\vspace{-0.75ex}
\[
W\le \frac{1}{\delta}=\frac{2}{\varepsilon^3}.\vspace{-0.75ex}
\]

Now consider any assignment
$
{\bf A}=(A_1,\ldots,A_L)
$
of the atomic units to the $L$ leaves of $\overline{\mathbf H}$, where
$A_j\subseteq [K]$ denotes the set of atomic units assigned to leaf $j$.
(We use ${\bf A}$ for the assignment of atomic units to distinguish it
from the assignment ${\bf S}$ of the original tokens.)
Define the exact and rounded masses of leaf $j$ by\vspace{-0.75ex}
\begin{align*}
m_j({\bf A})\triangleq \sum_{u\in A_j} q_u,
\qquad
\widehat m_j({\bf A})\triangleq \sum_{u\in A_j} \widehat q_u.\vspace{-0.75ex}
\end{align*}

Since $0\le q_u-\widehat q_u<\delta$ for every $u$, we have\vspace{-0.75ex}
\begin{align*}
\sum_{j=1}^L |m_j({\bf A})-\widehat m_j({\bf A})|
&=\sum_{j=1}^L \sum_{u \in A_j} \left( q_u - \widehat q_u \right)
=\sum_{u=1}^K (q_u-\widehat q_u)\\
&< K\delta \le \left(\frac{1}{\varepsilon^2}+1\right)\frac{\varepsilon^3}{2}
\hspace{-2.5pt}=\hspace{-2.5pt} \frac{\varepsilon}{2}\hspace{-0.8pt}+\hspace{-0.8pt}\frac{\varepsilon^3}{2}
\le \varepsilon,
\end{align*}
where the last inequality holds for $0<\varepsilon\le 1$.
Thus replacing the exact atomic masses by their rounded versions changes
the total leaf-mass vector by at most $\varepsilon$.

\paragraph{Dynamic program}
Although the assignment is to all $L$ leaves, the rounded mass of the
final leaf is determined by the first $L-1$ leaves because the total
rounded mass is fixed. For this reason, the dynamic program stores only
the first $L-1$ rounded leaf-loads.

Order the leaves as $1,2,\ldots,L$. For 
\[\vspace{-0.75ex}
\mathbf a=(a_1,\ldots,a_{L-1})\in \{0,1,\ldots,W\}^{L-1},
\]
interpret $a_j$ as the rounded load assigned to leaf $j$, measured in
units of $\delta$. Then the rounded load of leaf $L$ is implicitly
$
W-\sum_{j=1}^{L-1} a_j,
$
corresponding to a rounded mass
$
{\delta\cdot\big(W-\sum_{j=1}^{L-1} a_j\big)}.
$

We define a Boolean DP table
$
F_t(a_1,\ldots,a_{L-1}),
$
which is true if and only if the first $t$ atomic units can be assigned
to the $L$ leaves so that the rounded loads on the first $L-1$ leaves
are exactly $(a_1,\ldots,a_{L-1})$.
The initialization is
$
F_0(0,\ldots,0)=\mathrm{true},
$
and all other states are false. For atomic unit $t+1$, the transition
places it on one of the $L$ leaves:
\begin{itemize}
    \item if it is placed on leaf $j\in[L-1]$, then the $j$th coordinate
    increases by $w_{t+1}$;
    \item if it is placed on leaf $L$, then the state vector is unchanged.
\end{itemize}
For this setting we store the individual assignments of each of the
reachable states. For a reachable terminal state $\mathbf a$, the corresponding rounded
leaf masses are
\[
\delta a_1,\delta a_2,\ldots,\delta a_{L-1},
\quad
\delta\bigg(W-\sum_{j=1}^{L-1} a_j\bigg).
\]

Hence for any assignment ${\bf A}$ corresponding to the terminal state
$\mathbf a$, we define the rounded objective
\[
\widehat{\mathfrak D}({\bf A},\overline{\bf H})
\triangleq
\sum_{j=1}^{L-1}\left|2^{-\overline h_j}-\delta a_j\right|
+
\bigg|2^{-\overline h_L}
-\delta\bigg(W-\sum_{j=1}^{L-1}a_j\bigg)\bigg|.
\]
Moreover, by the reverse triangle inequality and the preceding bound on
the rounded leaf-mass error,\vspace{-0.75ex}
\begin{equation}
\label{eq:rounded-objective-gap}
\left|
\mathfrak D({\bf A},\overline{\bf H})-\widehat{\mathfrak D}({\bf A},\overline{\bf H})
\right|
\le
\sum_{j=1}^L |m_j({\bf A})-\widehat m_j({\bf A})|
\le
\varepsilon.\vspace{-0.75ex}
\end{equation}
The dynamic program returns a terminal state minimizing the rounded
objective, together with a corresponding explicit assignment of all
atomic units to all $L$ leaves.

\paragraph{Illustrative Example of the DP State Space}
To build intuition for the state-space compression and the transition function $F_t$, consider a minimal toy instance. Suppose we choose an error tolerance $\varepsilon = \sqrt[3]{0.002} \approx 0.126$. This comfortably bounds the required tree depth $d \le \log_2(1/\varepsilon) \approx 2.98$, allowing us to legally construct a tree with $L=3$ leaves at depths $h_1=1$, $h_2=2$, and $h_3=2$. The corresponding dyadic targets $2^{-h_j}$ are $0.50$, $0.25$, and $0.25$. 

We are given $K=3$ atomic units with true probability masses $m_1=0.50$, $m_2=0.30$, and $m_3=0.20$. Following the algorithm strictly, we apply the discretization step $\delta = \varepsilon^3/2 = 0.001$. The true masses are mapped to integer weights $w_u = \lfloor m_u / \delta \rfloor$, yielding $w_1=500$, $w_2=300$, and $w_3=200$. The fixed total integer weight is $W=1000$. The dyadic targets are analogously scaled by $1/\delta$ to establish the integer targets $T_1=500$, $T_2=250$, and $T_3=250$.

The dynamic program tracks the integer loads of the first $L-1=2$ leaves. A state $(a_1, a_2)$ is reachable after assigning $t$ items if the boolean function $F_t(a_1, a_2) = \mathrm{true}$. The DP initializes with $F_0(0,0) = \mathrm{true}$. As items are assigned, each reachable coordinate transitions to at most three new coordinates (adding $w_t$ to $a_1$, to $a_2$, or to neither if placed in Leaf 3). 

Figure~\ref{fig:dp-lattice} visualizes this state-space expansion. For visual clarity, the grid axes are scaled in units of 100. The colors represent the timeline of the expansion: $t=1$ (blue), $t=2$ (red), and $t=3$ (green).

\begin{figure}[htbp]
\centering
\begin{tikzpicture}[scale=0.6, every node/.style={font=\scriptsize}]
  \draw[step=1cm, gray!30, very thin] (-0.5,-0.5) grid (8.5,6.5);
  \draw[thick, ->] (-0.5,0) -- (8.5,0) node[right] {$a_1$ (Leaf 1, $\times 100$)};
  \draw[thick, ->] (0,-0.5) -- (0,6.5) node[above] {$a_2$ (Leaf 2, $\times 100$)};
  
  \foreach \x in {1,2,3,4,5,6,7,8} \draw (\x, 1pt) -- (\x, -3pt) node[anchor=north] {\x};
  \foreach \y in {1,2,3,4,5,6} \draw (1pt, \y) -- (-3pt, \y) node[anchor=east] {\y};

  \tikzset{
    st0/.style={circle, fill=black, inner sep=1.5pt},
    st1/.style={circle, fill=blue, inner sep=1.5pt},
    st2/.style={circle, fill=red, inner sep=1.5pt},
    st3/.style={circle, fill=green!60!black, inner sep=1.5pt},
    opt/.style={rectangle, fill=orange, draw=black, thick, inner sep=2.5pt},
    e1/.style={->, blue, thick, shorten >=2pt, shorten <=2pt},
    e2/.style={->, red, thick, shorten >=2pt, shorten <=2pt},
    e3/.style={->, green!60!black, thick, shorten >=2pt, shorten <=2pt}
  }

  \node[st0] (n00) at (0,0) {};
  \node[below left] at (n00) {$F_0$};

  \node[st1] (n50) at (5,0) {};
  \node[st1] (n05) at (0,5) {};
  \draw[e1] (n00) -- (n50);
  \draw[e1] (n00) -- (n05);
  \draw[e1, out=15, in=70, looseness=10] (n00) to (n00); 

  \node[st2] (n80) at (8,0) {};
  \node[st2] (n53) at (5,3) {};
  \draw[e2] (n50) -- (n80);
  \draw[e2] (n50) -- (n53);
  \draw[e2, out=15, in=70, looseness=10] (n50) to (n50); 
  
  \node[st3] (n73) at (7,3) {};
  \node[st3] (n55) at (5,5) {};
  \node[opt] (optNode) at (5,3) {}; 
  
  \draw[e3] (n53) -- (n73);
  \draw[e3] (n53) -- (n55);
    \draw[e3, out=15, in=70, looseness=10] (n53) to (n53); 
  \matrix [draw, fill=white, anchor=north west] at (0.5, 6.2) {
    \node[st1, label=right:Step 1 ($w_1{=}500$)] {}; \\
    \node[st2, label=right:Step 2 ($w_2{=}300$)] {}; \\
    \node[st3, label=right:Step 3 ($w_3{=}200$)] {}; \\
    \node[opt, label=right:Optimum] {}; \\
  };
\end{tikzpicture}
\vspace{-2ex}
\caption{2D DP state-space expansion for $L=3$. Arrows denote state transitions via placing an item in Leaf 1 (horizontal), Leaf 2 (vertical), or Leaf 3 (stationary). The final optimum rests at an integer coordinate of $(500,300)$.}
\label{fig:dp-lattice}
\end{figure}
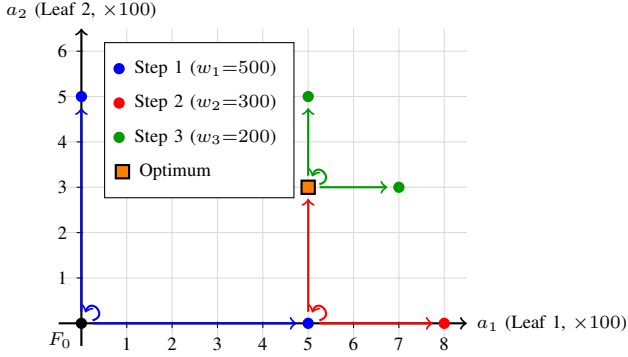

At the terminal step $t=K=3$, the algorithm evaluates the discrete divergence $\widehat{\mathfrak D}$ for \emph{all} valid terminal states. By implicitly defining the load of the final leaf as $W-\sum_{j=1}^{L-1}a_j$, the DP calculates the objective directly from the tracked coordinates:
\[
\widehat{\mathfrak D}({\bf A},\overline{\bf H})
\triangleq
\sum_{j=1}^{L-1}\left|2^{-h_j}-\delta a_j\right|
+
\bigg|2^{-h_L}
-\delta\bigg(W-\sum_{j=1}^{L-1}a_j\bigg)\bigg|
\]
Table~\ref{tab:dp-example} illustrates this exact calculation for the terminal nodes branching from $F_2(500,300)$. Evaluating this formula strictly identifies $(500, 300)$ as the optimal assignment.

\begin{table*}[htbp]
\centering
\caption{Evaluation of Selected Terminal States where $F_3(a_1, a_2) = \mathrm{true}$}
\label{tab:dp-example}
\renewcommand{\arraystretch}{1.3}
\begin{tabular}{@{}ccc@{}}
\toprule
\textbf{State $(a_1, a_2)$} & \textbf{Calculation for $\widehat{\mathfrak D}({\bf A},\overline{\bf H})$} & \textbf{Total $\widehat{\mathfrak D}$} \\
\midrule
$F_3(700, 300)$ & $|0.50{-}0.001(700)| + |0.25{-}0.001(300)| + |0.25{-}0.001(0)|$ & $0.50$ \\
$F_3(500, 500)$ & $|0.50{-}0.001(500)| + |0.25{-}0.001(500)| + |0.25{-}0.001(0)|$ & $0.50$ \\
$\mathbf{F_3(500, 300)}$ & $\mathbf{|0.50{-}0.001(500)| + |0.25{-}0.001(300)| + |0.25{-}0.001(200)|}$ & $\mathbf{0.10}$ \\
\bottomrule
\end{tabular}
\end{table*}

\begin{proposition}\label{prop:large-item-dp}
For every fixed candidate height vector $\overline{\mathbf H}$ of maximum
depth at most $\log_2(1/\varepsilon)$, one can compute an assignment of
the atomic units of the blocked instance to the leaves of $\overline{\mathbf H}$ in time\vspace{-0.75ex}
\[
O(n)\cdot
\exp\!\left(
O\!\left(\frac{\log(1/\varepsilon)}{\varepsilon}\right)
\right)\vspace{-0.75ex}
\]
whose true objective value is within $2\varepsilon$ of the optimum
blocked assignment for that fixed height vector, where temporary empty
leaves are permitted at this stage.
\end{proposition}
\begin{IEEEproof}
The number of DP states is at most
 \[
 (W+1)^{L-1}.
 \]
 Since $W\le 2\varepsilon^{-3}$ and $L\le \varepsilon^{-1}$, this gives
 \[
 (W+1)^{L-1}
 \le (2\varepsilon^{-3}+1)^{1/\varepsilon}
 =
 \exp\!\left( O\!\left(\frac{\log(1/\varepsilon)}{\varepsilon}\right)
\right).
 \]
 Each state has at most $L\le 1/\varepsilon$ outgoing transitions, and
 the number of atomic units satisfies
 \[
 K\le \frac{1}{\varepsilon^2}+1\le n.
 \]
 Hence the total running time is
 \[
 O\!\left(KL(W+1)^{L-1}\right)
 =
 O(n)\cdot
 \exp\!\left(
 O\!\left(\frac{\log(1/\varepsilon)}{\varepsilon}\right)
 \right).
 \]

 Let ${\bf A}^\ast$ be an optimal assignment for the blocked atomic instance under the true objective $\mathfrak D(\cdot,\overline{\bf H})$. For each $t\in\{0,\dots,K\}$, let
 \[
 a^{(t)}_j := \sum_{\substack{u\le t\\ u\in A^\ast_j}} w_u,
 \qquad j=1,\dots,L-1,
 \]
 denote the rounded loads on the first $L-1$ leaves induced by the first $t$ atomic units under ${\bf A}^\ast$. We claim that the DP state
 \[
 F_t\bigl(a^{(t)}_1,\dots,a^{(t)}_{L-1}\bigr)
 \]
 is reachable for every $t$. This is immediate by induction: $F_0(0,\dots,0)=\mathrm{true}$, and the transition for $t+1$ matches the assignment of atomic unit $t+1$ in ${\bf A}^\ast$. Hence the terminal rounded load vector induced by ${\bf A}^\ast$ is reachable. Since the DP returns a reachable terminal state minimizing $\widehat{\mathfrak D}(\cdot,\overline{\bf H})$, it follows that
 \[
 \widehat{\mathfrak D}({\bf A}_{\mathrm{DP}},\overline{\bf H}) \le \widehat{\mathfrak D}({\bf A}^\ast,\overline{\bf H}).
 \]
 Using \eqref{eq:rounded-objective-gap} for both ${\bf A}_{\mathrm{DP}}$ and ${\bf A}^\ast$, we conclude
 \begin{align*}
 \mathfrak D({\bf A}_{\mathrm{DP}},\overline{\bf H})
 &\le \widehat{\mathfrak D}({\bf A}_{\mathrm{DP}},\overline{\bf H})+\varepsilon \\
 &\le \widehat{\mathfrak D}({\bf A}^\ast,\overline{\bf H})+\varepsilon \\
 &\le \mathfrak D({\bf A}^\ast,\overline{\bf H})+2\varepsilon.
 \end{align*}
 \end{IEEEproof}

\subsection{Feasibility Repair}\label{sec:repair}

The truncation reduction (Lemma~\ref{lem:truncation}) preserves divergence, but a depth-truncated candidate may have $\mathfrak{R}<R$. Whenever such a candidate has a leaf at depth exactly $d$, we can supply the missing rate by replacing that leaf with a complete subtree of $T_L\triangleq 2^{R/\varepsilon}$ uniform sub-leaves. Small items from $\cS_\text{small}$ play two roles: filling each new sub-leaf and absorbing the migrated mass. The construction is cheap precisely because the deepest leaf can be chosen to carry the smallest mass, making its dyadic-target perturbation small.

\begin{lemma}[Monotone optimal ordering]\label{lem: monotonicity}
Fix a height vector $\bf H$, and reindex the leaves so that
$h_1\le h_2\le \cdots \le h_L.$
Then there exists an optimal partition $\bf S^*$ minimizing $\mathfrak{D}({\bf S},{\bf H})$ 
such that\vspace{-0.75ex}
\[
\Pr(\cS_1^*)\ge \Pr(\cS_2^*)\ge \cdots \ge \Pr(\cS_L^*).\vspace{-0.75ex}
\]
\end{lemma}
 \begin{IEEEproof}
 Let $x_j\triangleq 2^{-h_j}$, so $x_1\ge x_2\ge \cdots \ge x_L$. Take any optimal partition. If there exist indices $i<j$ with $\Pr(\cS_i)<\Pr(\cS_j)$, set
 \[
 a\triangleq \Pr(\cS_i),\ \ b\triangleq \Pr(\cS_j),\ \ x\triangleq x_i,\ \ y\triangleq x_j.
 \]
 Since $a\le b$ and $x\ge y$, the function $\phi(z)\triangleq |z-x|-|z-y|$ is nonincreasing in $z$, and so
 \[
 |a-x|+|b-y|-|a-y|-|b-x|=\phi(a)-\phi(b)\ge 0.
 \]
 Hence swapping the assignments of $\cS_i$ and $\cS_j$ does not increase the objective. Repeating finitely often yields an optimal partition with nonincreasing leaf masses.
 \end{IEEEproof}

 Therefore, we can assume that the deepest leaf carry the smallest mass. 
 In particular, when its mass 
 approximates the dyadic target $\varepsilon=2^{-d}$, the subtree replacement contributes only $O(\varepsilon)$ to the divergence, as formalized in the next lemma. 

\begin{lemma}[Feasibility repair]\label{lem:repair}
Let $\overline{\bf H}$ be a height vector of maximum depth $d=\log_2(1/\varepsilon)$ with $L\le 1/\varepsilon$ leaves and at least one leaf at depth exactly $d$, and let ${\bf S}$ be a surjective partition for $\overline{\bf H}$. Suppose Assumptions~1--3 hold and that there is a subset $\cS\subseteq \cS_\text{small}$ of cardinality $T_L\triangleq 2^{R/\varepsilon}$. Then there exists $(\tilde{\bf S},\tilde{\bf H})$ feasible for Problem~1 (i.e., $\mathfrak{R}(\tilde{\bf S},\tilde{\bf H})\ge R$) with
\[
\mathfrak D(\tilde{\bf S},\tilde{\bf H})\le \mathfrak D({\bf S},\overline{\bf H})+4\varepsilon.
\]
\end{lemma}

\begin{IEEEproof}
 By Lemma~\ref{lem: monotonicity} we may reindex the leaves of $\overline{\bf H}$ so that
 \[
 \overline h_1\le \cdots\le \overline h_L=d
 \quad\text{and}\quad
 \Pr(\cS_1)\ge \cdots \ge \Pr(\cS_L).
 \]
 The deepest leaf $L$ has dyadic target mass $2^{-d}=\varepsilon$.

 \emph{Migration.} Each item in $\cS_\text{small}$ has mass at most $\varepsilon\cdot 2^{-R/\varepsilon}=\varepsilon/T_L$, so $\Pr(\cS)\le T_L\cdot\varepsilon/T_L=\varepsilon$. Move every element of $\cS$ to leaf $L$ to obtain a new partition ${\bf S}''$. Each individual move changes the mass of at most two leaves, so
 \begin{align}\label{eq:repair-mig}
 \left|\mathfrak D({\bf S}'',\overline{\bf H})-\mathfrak D({\bf S},\overline{\bf H})\right|
 \le 2\Pr(\cS)\le 2\varepsilon.
 \end{align}
 Set $m\triangleq \Pr({\bf S}''_L)$ and $\beta\triangleq |m-\varepsilon|$.

 \emph{Expansion.} Replace leaf $L$ with a complete binary subtree of $T_L$ leaves. Each new leaf has dyadic target $t\triangleq \varepsilon/T_L$. Write $\cS=\{i_1,\ldots,i_{T_L}\}$ and $q_r\triangleq p_{i_r}\le t$. Assign item $i_r$ to leaf $r\in[T_L]$, and assign all elements of ${\bf S}''_L\setminus \cS$ to leaf $1$. Let $(\tilde{\bf S},\tilde{\bf H})$ denote the resulting tree, and define
\[
 \eta\triangleq \sum_{r=2}^{T_L}(t-q_r)\in\bigl[0,(T_L-1)t\bigr)\subseteq[0,\varepsilon).
 \]
 For $r\neq 1$, leaf $r$ contributes $|q_r-t|=t-q_r$ to the divergence, while leaf $1$ has mass $m-\sum_{r=2}^{T_L} q_r$ and contributes $|(m-\varepsilon)+\eta|$. The total contribution of the expanded subtree is therefore $|(m-\varepsilon)+\eta|+\eta$, while the original contribution of leaf $L$ was $\beta=|m-\varepsilon|$. Hence
 \begin{align}\label{eq:repair-exp}
\bigl|\mathfrak D(\tilde{\bf S},\tilde{\bf H})-\mathfrak D({\bf S}'',\overline{\bf H})\bigr|
 &\le \bigl||(m-\varepsilon)+\eta|-|m-\varepsilon|\bigr|+\eta\nonumber\\
 &\le 2\eta< 2\varepsilon.
 \end{align}
 Combining \eqref{eq:repair-mig} and \eqref{eq:repair-exp},
 \[
 \mathfrak D(\tilde{\bf S},\tilde{\bf H})\le \mathfrak D({\bf S},\overline{\bf H})+4\varepsilon.
 \]

 \emph{Rate.} The expanded subtree alone contributes
 \begin{align*}
 T_L\cdot 2^{-(d+\log_2 T_L)}\cdot (d+\log_2 T_L)
 &=\varepsilon\cdot (d+\log_2 T_L) \\
 &\ge \varepsilon\cdot \log_2 T_L \;=\; R
 \end{align*}
 to $\mathfrak{R}(\tilde{\bf S},\tilde{\bf H})$. The remaining leaves contribute non-negatively, so $\mathfrak{R}(\tilde{\bf S},\tilde{\bf H})\ge R$.
 \end{IEEEproof}

\subsection{The Algorithm and Main Theorem}\label{sec:main-theorem}

The atomic-assignment DP (Proposition~\ref{prop:large-item-dp}) returns an assignment that may leave some leaves empty, since the DP works over rounded mass-vectors and does not enforce one-item-per-leaf. Problem~1 forbids empty leaves, so before we can claim feasibility, surjectivity has to be restored. Since small items are abundant, we force one designated small item into each leaf. The cost is small because each small item has mass at most $\varepsilon\cdot 2^{-R/\varepsilon}$, which is a constant times $\varepsilon$ when summed over the $L\le 1/\varepsilon$ leaves.

\begin{lemma}[Leaf seeding with reserved small items]\label{lem:seed}
Fix a height vector $\overline{\bf H}$ with $L \leq \frac{1}{\varepsilon}$ leaves, and let
$r_1,\ldots,r_L$ be distinct elements of $\cS_{\mathrm{small}}$.
For any assignment ${\bf S}=(\cS_1,\ldots,\cS_L)$ of the support to the
leaves of $\overline{\bf H}$, there exists another assignment
${\bf S}^{\mathrm{seed}}$ such that $r_j\in \cS_j^{\mathrm{seed}}$ for
every $j\in[L]$, and \vspace{-0.75ex}
\[
\left|
\mathfrak D({\bf S}^{\mathrm{seed}},\overline{\bf H})
-
\mathfrak D({\bf S},\overline{\bf H})
\right|
\le
2\sum_{j=1}^L p_{r_j}
\le
2L\varepsilon\,2^{-R/\varepsilon}
\le
2\varepsilon.\vspace{-0.75ex}
\]
In particular, ${\bf S}^{\mathrm{seed}}$ is surjective.
\end{lemma}

 \begin{IEEEproof}
 For each $j\in[L]$, move the item $r_j$ to leaf $j$ (do nothing if it is already there). Moving $r_j$ can change the mass of at most two leaves --- its source leaf and leaf $j$ --- and therefore changes the divergence by at most $2p_{r_j}$. Summing over $j$,
 \[
 \left|
 \mathfrak D({\bf S}^{\mathrm{seed}},\overline{\bf H})
 -
 \mathfrak D({\bf S},\overline{\bf H})
 \right|
 \le 2\sum_{j=1}^L p_{r_j}.
 \]
 Since each $r_j\in\cS_{\mathrm{small}}$ has $p_{r_j}\le \varepsilon\,2^{-R/\varepsilon}$,
 \[
 2\sum_{j=1}^L p_{r_j}
 \le 2L\varepsilon\,2^{-R/\varepsilon}\le 2\varepsilon,
 \]
 where the last inequality uses $L\le 1/\varepsilon$ and Assumption~3. Every leaf contains its designated seed, so the assignment is surjective.
 \end{IEEEproof}
\medskip

\noindent The four ingredients (truncation, blocking, DP, repair), together with the seeding lemma, give the algorithm below.

\medskip

\noindent\textbf{Algorithm.} Enumerate every full binary tree shape $\overline{\bf H}$ of maximum depth $d=\log_2(1/\varepsilon)$. For each:
\begin{itemize}
\item[(a)] \emph{Block} items of mass below $\varepsilon^2$ into atomic units (Lemma~\ref{lem:blocking}).
\item[(b)] Run the \emph{atomic-assignment DP} (Proposition~\ref{prop:large-item-dp}) for $\overline{\bf H}$, returning an assignment ${\bf S}_{\overline{\bf H}}$ of the atomic units, possibly with empty leaves.
\item[(c)] Unpack each block into its constituent items at the leaf to which the block was assigned, then \emph{seed} ${\bf S}_{\overline{\bf H}}$ with $L$ reserved small items (Lemma~\ref{lem:seed}) to obtain a surjective ${\bf S}_{\overline{\bf H}}^{\mathrm{seed}}$.
\item[(d)] If $\mathfrak{R}(\overline{\bf H})\ge R$, accept $({\bf S}_{\overline{\bf H}}^{\mathrm{seed}},\overline{\bf H})$. Else if $\overline{\bf H}$ has a depth-$d$ leaf, apply \emph{repair} (Lemma~\ref{lem:repair}) using $T_L=2^{R/\varepsilon}$ further small items disjoint from the seeds, and accept the result. Else discard.
\end{itemize}
Return the accepted candidate of minimum divergence.

\begin{theorem}\label{thm:main}
Suppose $R$ is constant with respect to $n$ and that Assumptions~1--3 hold. The algorithm above runs in time
\[
O(n)\cdot \exp\!\left(O\!\left(\frac{\log(1/\varepsilon)}{\varepsilon}\right)\right)
\]
and outputs a feasible solution $(\tilde{\bf S},\tilde{\bf H})$ to Problem~1 with
\[
\mathfrak D(\tilde{\bf S},\tilde{\bf H})\le \mathrm{OPT}+12\varepsilon.
\]
In particular, for fixed $\varepsilon$ and constant $R$, the running time is linear in $n$.
\end{theorem}

\begin{IEEEproof}
 Let $({\bf S}^*,{\bf H}^*)$ be any optimal solution, with $\mathfrak D({\bf S}^*,{\bf H}^*)=\mathrm{OPT}$.

 \textit{Per-candidate divergence bound.} Fix a candidate $\overline{\bf H}$ of depth at most $d$. The blocked instance for $\overline{\bf H}$ has optimum value at most $\mathrm{OPT}(\overline{\bf H})+4\varepsilon$ (Lemma~\ref{lem:blocking}). Step (b) returns an assignment within $2\varepsilon$ of that blocked optimum (Proposition~\ref{prop:large-item-dp}), and unpacking blocks does not change leaf masses. Hence after step (b),
 \begin{align}\label{eq:thm-after-dp}
 \mathfrak D({\bf S}_{\overline{\bf H}},\overline{\bf H})\le \mathrm{OPT}(\overline{\bf H})+6\varepsilon.
 \end{align}
 Step (c) seeds at additional cost $2\varepsilon$ (Lemma~\ref{lem:seed}), giving
 \begin{align}\label{eq:thm-after-seed}
 \mathfrak D({\bf S}_{\overline{\bf H}}^{\mathrm{seed}},\overline{\bf H})\le \mathrm{OPT}(\overline{\bf H})+8\varepsilon.
 \end{align}
 If step (d) accepts at the first branch, this is the produced bound. Otherwise repair adds at most $4\varepsilon$ (Lemma~\ref{lem:repair}), giving
 \begin{align}\label{eq:thm-after-repair}
 \mathfrak D(\tilde{\bf S},\tilde{\bf H})\le \mathrm{OPT}(\overline{\bf H})+12\varepsilon.
 \end{align}

 \textit{Global bound.} By Lemma~\ref{lem:truncation} there is a candidate $\overline{\bf H}^*$ of depth $\le d$ with $\mathrm{OPT}(\overline{\bf H}^*)\le \mathrm{OPT}$, satisfying:
 \begin{itemize}
 \item if $\max_j h^*_j\le d$, then $\overline{\bf H}^*={\bf H}^*$ and $\mathfrak{R}(\overline{\bf H}^*)\ge R$, so the first branch of step (d) applies and \eqref{eq:thm-after-seed} gives divergence $\le \mathrm{OPT}+8\varepsilon$;
 \item if $\max_j h^*_j>d$, then $\overline{\bf H}^*$ has a depth-$d$ leaf (Lemma~\ref{lem:truncation}), so the repair branch applies and \eqref{eq:thm-after-repair} gives divergence $\le \mathrm{OPT}+12\varepsilon$.
 \end{itemize}
 Either way, the candidate $\overline{\bf H}^*$ produces a feasible solution within $\mathrm{OPT}+12\varepsilon$, and the algorithm returns the best feasible candidate.

 \textit{Running time.} Blocking takes $O(n)$ time. The number of bounded-depth tree shapes is at most $2^{2^{d+1}-1}=\exp(O(1/\varepsilon))$. For each shape, steps (b)--(d) run in time $O(n)\cdot \exp(O(\log(1/\varepsilon)/\varepsilon))$, dominating $\exp(O(1/\varepsilon))$. Total running time is $O(n)\cdot \exp(O(\log(1/\varepsilon)/\varepsilon))$.
 \end{IEEEproof}

\section{Discussion}\label{sec:discussion}

\subsection{The Dual Formulation}

The paper focuses on Problem~1 (minimize divergence subject to a rate constraint). The dual problem --- Problem~2, maximize rate subject to a divergence constraint --- is equally natural and the same machinery applies, with truncation, blocking, and the bounded-depth DP carrying over without change. Only the feasibility check and the role of the repair lemma swap places. We omit the details.

\subsection{Scaling $R$ or $L$ with $n$}

Both $R$ and $L$ are constants in our setting. When $R = \Theta(\log n)$, the number of bounded-depth tree shapes is no longer a constant in $n$, and tighter enumeration strategies are required. We expect approximate counting and structured enumeration ideas from the polynomial-knapsack literature~\cite{Kellerer2004} to be relevant. When $L$ scales with $n$, the problem moves toward a near-symbol-wise regime where vanishing divergence is achievable by trivial constructions. Neither regime is immediately motivated by LLM steganography, but they raise the question of whether the constant in $\mathrm{OPT}+12\varepsilon$ depends qualitatively on the regime.





\begin{thebibliography}{20}

\bibitem{Alon1998}
N.~Alon, Y.~Azar, G.~J.~Woeginger, and T.~Yadid,
``Approximation schemes for scheduling on parallel machines,''
\emph{Journal of Scheduling},
vol.~1, no.~1, pp.~55--66, 1998.

\bibitem{Bloch2016}
M.~Bloch and J.~Barros,
\emph{Physical-Layer Security: From Information Theory to Security Engineering}.
Cambridge, UK: Cambridge University Press, 2016.

\bibitem{Buchem2020}
M.~Buchem and L.~Rohwedder,
``Additive approximation schemes for load balancing problems,''
\emph{arXiv preprint arXiv:2007.09333}, 2020.

\bibitem{Cachin1998}
C.~Cachin,
``An information-theoretic model for steganography,''
in \emph{Proc. 2nd Int. Workshop on Information Hiding (IH)},
Portland, OR, USA, 1998, pp.~306--318.

\bibitem{CapraraMSSP2000}
A.~Caprara, H.~Kellerer, and U.~Pferschy,
``A PTAS for the multiple subset sum problem with different knapsack capacities,''
\emph{Information Processing Letters},
vol.~73, no.~3--4, pp.~111--118, 2000.

\bibitem{ChakrabartySwamy2019}
D.~Chakrabarty and C.~Swamy,
``Approximation algorithms for minimum norm and ordered optimization problems,''
in \emph{Proc. 51st ACM Symposium on Theory of Computing (STOC)},
Phoenix, AZ, USA, 2019, pp.~126--137.

\bibitem{ChekuriKhanna2006}
C.~Chekuri and S.~Khanna,
``A polynomial time approximation scheme for the multiple knapsack problem,''
\emph{SIAM Journal on Computing},
vol.~35, no.~3, pp.~713--728, 2005.

\bibitem{CoverThomas2006}
T.~M.~Cover and J.~A.~Thomas,
\emph{Elements of Information Theory}, 2nd~ed.
Hoboken, NJ, USA: Wiley-Interscience, 2006.

\bibitem{Dai2019}
F.~Dai, Y.~Zhang, and D.~Wang,
``Towards diverse and natural image descriptions via a conditional GAN,''
in \emph{Proc. Adv. Neural Inf. Process. Syst. (NeurIPS)},
Vancouver, BC, Canada, 2019.

\bibitem{Fridrich2009}
J.~Fridrich,
\emph{Steganography in Digital Media: Principles, Algorithms, and Applications}.
Cambridge, UK: Cambridge University Press, 2009.

\bibitem{GareyJohnson1979}
M.~R.~Garey and D.~S.~Johnson,
\emph{Computers and Intractability: A Guide to the Theory of NP-Completeness}.
New York, NY, USA: W.~H.~Freeman, 1979.

\bibitem{Glover1989}
F.~Glover,
``Tabu search---Part I,''
\emph{ORSA Journal on Computing},
vol.~1, no.~3, pp.~190--206, 1989.

\bibitem{Huang2024ODStega}
Y.-S.~Huang, P.~Just, K.~Narayanan, and C.~Tian,
``OD-Stega: LLM-based near-imperceptible steganography via optimized distributions,''
\emph{arXiv preprint arXiv:2410.04328}, Oct.~2024.

\bibitem{IbarraKim1975}
O.~H.~Ibarra and C.~E.~Kim,
``Fast approximation algorithms for the knapsack and sum of subset problems,''
\emph{Journal of the ACM},
vol.~22, no.~4, pp.~463--468, 1975.

\bibitem{Kellerer2004}
H.~Kellerer, U.~Pferschy, and D.~Pisinger,
\emph{Knapsack Problems}.
Berlin, Germany: Springer, 2004.

\bibitem{Korf2009}
R.~E.~Korf,
``Multi-way number partitioning,''
in \emph{Proc. 21st Int. Joint Conf. Artificial Intelligence (IJCAI)},
Pasadena, CA, USA, 2009, pp.~538--543.

\bibitem{Lawler1979}
E.~L.~Lawler,
``Fast approximation algorithms for knapsack problems,''
\emph{Mathematics of Operations Research},
vol.~4, no.~4, pp.~339--356, 1979.

\bibitem{MartelloToth1990}
S.~Martello and P.~Toth,
\emph{Knapsack Problems: Algorithms and Computer Implementations}.
Chichester, U.K.: Wiley, 1990.

\bibitem{Sahni1975}
S.~Sahni,
``Approximate algorithms for the 0/1 knapsack problem,''
\emph{Journal of the ACM},
vol.~22, no.~1, pp.~115--124, 1975.

\bibitem{Ziegler2019}
Z.~M.~Ziegler and A.~M.~Rush,
``Neural text degeneration with unlikelihood training,''
in \emph{Proc. Int. Conf. Learn. Represent. (ICLR)},
New Orleans, LA, USA, 2019.
\end{thebibliography}


\end{document}